\documentclass[graybox]{svmult}

\usepackage{algorithm,algorithmic,amsbsy,amsfonts,amsmath,amssymb,bbm,bm,booktabs,cite,color,courier,epstopdf,graphicx,import}
\usepackage{helvet,makeidx,mathptmx,mathtools,multicol,subfigure,slashed,times,transparent,type1cm,url,verbatim}
\usepackage[USenglish]{babel}
\usepackage[T1]{fontenc}
\usepackage[utf8]{inputenc}
\usepackage[nolist]{acronym}
\usepackage[bottom]{footmisc}

{\newtheorem{remarkb}{Remark}}

\renewcommand{\E}{\mathbf{E}}

\renewcommand{\I}{\mathbf{I}}

\newcommand{\setF}{\mathcal{F}}

\newcommand{\setL}{\mathcal{L}}

\newcommand{\diff}{\mathrm{d}}
\newcommand{\Exp}{\mathbb{E}}

\renewcommand{\sc}{\textnormal{\tiny{SBS}}}
\newcommand{\ul}{\textnormal{\tiny{UL}}}
\newcommand{\dl}{\textnormal{\tiny{DL}}}
\newcommand{\sir}{\mathsf{SIR}}
\newcommand{\Phit}{\mathsf{P}_{\mathrm{hit}}}
\newcommand{\Psuc}{\mathsf{P}_{\mathrm{suc}}}
\newcommand{\PsucLB}{\underline{\mathsf{P}}_{\mathrm{suc}}}

\newcommand{\red}[1]{\textcolor{black}{#1}}

\makeindex

\begin{document}

%=========================================================================
%\include{acronyms}
\newacro{3GPP}{3rd Generation Partnership Project}
\newacro{4G}{fourth generation}
\newacro{5G}{fifth generation}
\newacro{ASE}{area spectral efficiency}
\newacro{AWGN}{additive white Gaussian noise}
\newacro{BB}{base-band}
\newacro{BC}{broadcast channel}
\newacro{BD}{block diagonalization}
\newacro{BER}{bit error rate}
\newacro{BN}{backhaul node}
\newacroplural{BNs}{backhaul nodes}
\newacro{BS}{base station}
\newacroplural{BSs}{base stations}
\newacro{CAPEX}{capital expense}
\newacro{CDF}{cumulative distribution function}
\newacro{CF}{collaborative filtering}
\newacro{CQI}{channel quality indicator}
\newacro{CSI}{channel state information}
\newacro{CSIT}{channel state information at the transmitter}
\newacro{D2D}{device-to-device}
\newacro{DL}{downlink}
\newacro{DPC}{dirty paper coding}
\newacro{FD}{full-duplex}
\newacro{FDD}{frequency-division duplexing}
\newacro{FDMA}{frequency division multiple access}
\newacro{HD}{half-duplex}
\newacro{HetNet}{heterogeneous network}
\newacroplural{HetNets}{hetrogenoeus networks}
\newacro{IA}{interference alignment}
\newacro{i.i.d.}{independent and identically distributed}
\newacro{INI}{inter-node interference}
\newacro{KKT}{Karush-Kuhn-Tucker}
\newacro{LOS}{line-of-sight}
\newacro{LTE}{long term evolution}
\newacro{M2M}{machine-to-machine}
\newacro{MCS}{modulation and coding scheme}
\newacro{MIMO}{multiple-input multiple-output}
\newacro{MMSE}{minimum mean square error}
\newacro{mmWave}{millimeter wave}
\newacro{MRC}{maximum-ratio combining}
\newacro{MRT}{maximum-ratio transmission}
\newacro{MU}{multiuser}
\newacro{NLOS}{non-line-of-sight}
\newacro{NR}{New Radio}
\newacro{OPEX}{operating expenses}
\newacro{PPP}{Poisson point process}
\newacroplural{PPPs}{Poisson point processes}
\newacro{QoS}{quality of service}
\newacro{RAN}{radio access network}
\newacro{RF}{radio frequency}
\newacro{RX}{receiver}
\newacroplural{RXs}{receivers}
\newacro{SA}{standalone}
\newacro{SBS}{small-cell base station}
\newacroplural{SBSs}{small-cell \acp{BS}}
\newacro{SC}{small cell}
\newacroplural{SCs}{small cells}
\newacro{SI}{self-interference}
\newacro{SIC}{self-interference cancellation}
\newacro{SINR}{signal-to-interference-plus-noise ratio}
\newacro{SIR}{signal-to-interference ratio}
\newacro{SNR}{signal-to-noise ratio}
\newacro{SU}{single user}
\newacro{SVD}{singular value decomposition}
\newacro{SWIPT}{simultaneous wireless information and power transfer}
\newacro{TDD}{time-division duplexing}
\newacro{TDMA}{time division multiple access}
\newacro{TL}{transfer learning}
\newacro{TX}{transmitter}
\newacroplural{TXs}{transmitters}
\newacro{UDN}{ultra-dense network}
\newacro{UE}{user equipment}
\newacroplural{UEs}{user equipments}
\newacro{UL}{uplink}
\newacro{ULA}{uniform linear array}
\newacroplural{ULAs}{uniform linear arrays}
\newacro{UT}{user terminal}
\newacroplural{UTs}{user terminals}
\newacro{WEC}{wireless edge caching}
\newacro{WF}{water-filling}
\newacro{WPT}{wireless power transfer}
\newacro{ZF}{zero forcing}
%=========================================================================

\title*{Full-Duplex Radios for Edge Caching}
\titlerunning{Full-Duplex Radios for Edge Caching}
\author{Italo Atzeni and Marco Maso}
\institute{Italo Atzeni \at Communication Systems Department, EURECOM, Sophia Antipolis, France. \\
\email{italo.atzeni@ieee.org} \\
Marco Maso\at Nokia Bell Labs, Paris-Saclay, France. \\
\email{marco.maso@nokia.com}}

\maketitle

\vspace{-2cm}

\setcounter{page}{0}
\pagenumbering{arabic}
\setcounter{page}{1}

%=========================================================================
%Abstract
%=========================================================================
	
\abstract{\red{Recent studies have shown that edge caching may have a beneficial effect on the sustainability of future wireless networks}. While its positive impact at the network level is rather clear (in terms of, e.g., access delay and backhaul load), assessing its potential benefits at the physical layer is less straightforward. This chapter builds upon this observation and \red{focuses} on the performance enhancement brought by the addition of caching capabilities to \ac{FD} radios in the context of \acp{UDN}. \red{More specifically, we aim at showing that} the interference footprint of such networks, i.e., the major bottleneck to overcome to observe the theoretical \ac{FD} throughput doubling at the network level, can be significantly reduced thanks to edge caching. As a matter of fact, fundamental results available in the literature show that most of the gain, as compared to their \ac{HD} counterparts, can be achieved by such networks only if costly modifications to their infrastructure are performed and/or if high-rate signaling is exchanged between \acp{UE} over suitable control links. Therefore, we aim at proposing a viable and cost-effective alternative to these solutions based on \red{pre-fetching locally popular contents at the network edge}. We start by considering an interference-rich scenario such as an ultra-dense \ac{FD} small-cell network, in which several non-cooperative \ac{FD} base stations (BSs) serve their associated \acp{UE} while communicating with a wireless \ac{BN} to retrieve the content to deliver. We then describe \red{a geographical caching policy aiming at capturing local files popularity and compute the corresponding cache-hit probability. Thereupon}, we calculate the probability of successful transmission of a file requested by a \ac{UE}, either directly by its serving \ac{SBS} or by the corresponding \ac{BN}: this quantity is then used to lower-bound the throughput of the considered network. \red{Our approach leverages tools from stochastic geometry in order to guarantee both analytical tractability of the problem and generality of the results}. A set of suitable numerical simulations is finally performed to confirm the correctness of the theoretical findings and characterize the performance enhancement brought by the adoption of edge caching. The most striking result in this sense is the remarkable performance improvement observed when \red{shifting from cache-free to cache-aided \ac{FD} small-cell networks}.}

%=========================================================================
\section{Introduction} \label{sec:INTRO}
%=========================================================================

The last decade has witnessed the progressive introduction of the \ac{4G} cellular network technology and the concurrent adoption of increasingly competitive pricing strategies by device manufacturers and telcos. As a consequence, devices that are able to offer reliable broadband data connections to their users, i.e., smartphones, ceased to be premium products and became a commodity. Their market penetration is already massive and keeps progressing at steady pace. Recent studies forecast that smartphones will represent 86$\%$ of the total mobile data traffic by 2021, compared to 81$\%$ in 2016, and that \red{monthly} mobile data traffic will reach 49 exabytes \red{worldwide} (or\red{, equivalently,} a run rate of 587 exabytes annually) \cite{Cis17}.

The amount of network resources needed to support these trends is ever-increasing. Telcos already anticipate that current mobile networks will have to be restructured to cope with both future service demands and the multitude of novel mobile broadband applications constantly introduced in the market. Many important requirements have been identified in this context, such as the need for higher spectral and energy efficiency, lower end-to-end delays, better coverage, large scalability, and lower \ac{CAPEX} and \ac{OPEX}, just to name a few \cite{And14}. As a consequence, one of the strongest drivers in the last years for several research groups in both industry and academia has been the need to define a more advanced and flexible network technology as compared to \ac{4G}, i.e., the so-called \ac{5G}. The remarkable results of such activities have already yielded significant outcomes within \red{standardization development organizations} like the \ac{3GPP}, who have already published the first version of the standard that will guide the deployment of future \ac{5G} wireless networks, i.e., \ac{3GPP} Release 15 \cite{3GPP38211,3GPP38214}. This evolution, conventionally referred to as \ac{NR}, is planned to hit the market in its non-standalone version within 2019, and is expected to provide significant gains over previous systems. 

From a practical point of view, \ac{NR} deployments will be characterized by the introduction, or further development, of several key solutions expected to bring the sought performance enhancement as compared to existing networks. Interestingly, only some strategies and network configurations have been and are subject to standardization, whereas some others are considered as part of the implementation aspects. Noteworthy and representative examples of these two categories are \cite{3GPP38211,3GPP38214}:
\begin{itemize}
	\item[$\bullet$] \textit{Massive \ac{MIMO}}: this natural candidate for the physical layer of NR has imposed a revision of the reference sequences and \ac{CSI} feedback mechanisms \cite{Boc14,Bjo17};
	\item[$\bullet$] \textit{Advanced \ac{MIMO} precoding}: the adoption of such strategies at the \ac{BS} should be completely transparent to the \ac{UE}, i.e., precoding solutions are implementation aspects that are not specified in the standard. 
\end{itemize}
As a matter of fact, the relevance and impact of many other technologies and approaches will increase in future \ac{5G} networks as compared to their current role in mobile and fixed networks, regardless of their \ac{3GPP} standardization status (i.e., specified or not). In this chapter, we specifically focus on two of these approaches to study and discuss the potential brought by their mutual interactions:
\begin{itemize}
	\item[$\bullet$] The so-called \red{\textit{proactive caching} at the network edge}, by means of which contents (e.g., videos, images, and news) are brought closer to the users and intelligently cached at \acp{SBS} equipped with high-capacity storage units, as illustrated in Fig.~\ref{fig:caching}. As a result, the end-to-end access delay is significantly reduced, the mobile infrastructure is offloaded, and the impact of limited-capacity backhaul on the network performance is mitigated \cite{Bas15,Li18,Ham16,Yan18,Kri17,Pas16}. The role of edge caching becomes particularly crucial in case of \ac{UDN} deployments, i.e., massively populated (and possibly heterogeneous) networks in which the distance between \acp{BS} and served \acp{UE} is reduced as compared to classic macro-cell networks \cite{Liu17,Khr16,Mas17,Atz17a}. Such \acp{UDN} may comprise several layers, each of them including different \red{categories of cells (i.e., femto, pico, micro, and macro cells)} \cite{And11,Dhi12,Yun15}. In general, this layered architecture allows to design efficient strategies to offload the pre-existing macro-cell infrastructure and enhance the network capacity, especially when several nodes provide caching support \cite{Bas14}.
	\item[$\bullet$] The transition from \ac{HD} to \textit{\ac{FD} operations} at radio terminals also promises to offer many benefits, although subject to some peculiar limitations \cite{Sab14,Atz15,Ton15,Goy15,Atz17}. \red{A \ac{FD} device does not require separate time/frequency resources to be able to support data transmission and reception. In other words, it can simultaneously transmit and receive data over the same bandwidth, thus having the potential to achieve a theoretical throughput doubling and energy efficiency enhancement in comparison to \ac{HD} radios}. In particular, equipping network nodes with \ac{FD} capabilities can simplify the adoption of flexible duplexing strategies such as dynamic \ac{TDD} and enable readjustments to frame structures on-the-fly. Additionally, \ac{FD} transmission offers advantages in terms of operation, cost, and efficiency as compared to traditional \ac{HD} operating mode \cite{Mas15}.
\end{itemize}
\begin{figure*}[t!]
	\centering
	\includegraphics[scale=0.32]{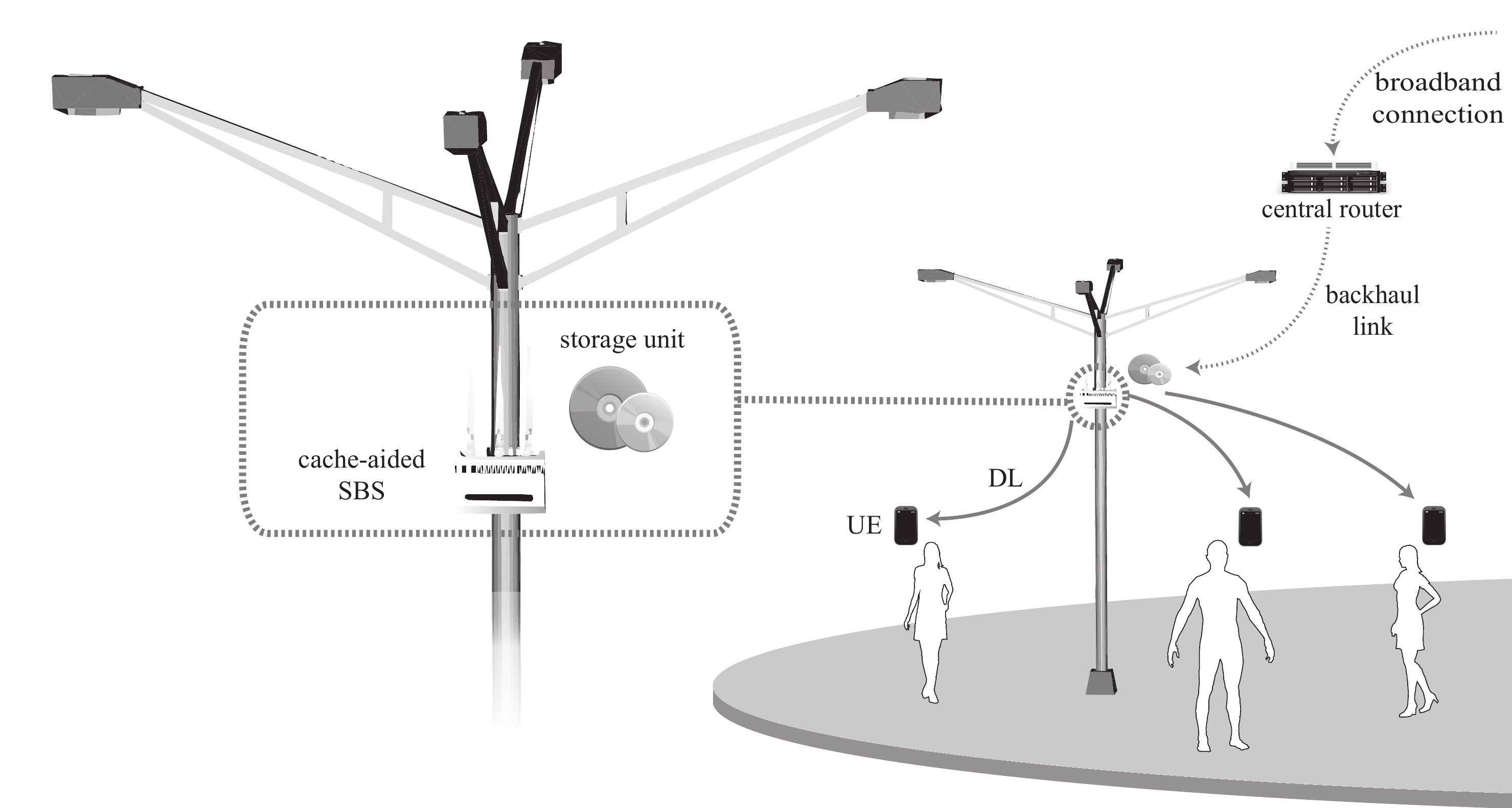}
	\caption{Cache-aided \ac{SBS}: the \ac{SBS} is equipped with a storage unit to pre-fetch popular content.}
	\label{fig:caching}
\end{figure*} 
\red{The aforementioned approaches certainly have significant potential if taken individually. Nevertheless, assessing the extent of their interoperability is not straightforward}. This is mostly due to the interference footprint of the \ac{FD} links \cite{Goy15}, which may complicate a seamless integration of caching capabilities at each network nodes. At this stage, a brief introduction of such technology is in order, to better characterize its features and issues, before studying the impact of edge caching on the performance of \ac{FD} radios and networks.

%=========================================================================
\subsection{Full Duplex Communications} \label{sec:FullDuplexCommunications}
%=========================================================================

The majority of current wireless radios operate in \ac{HD} mode. In practice, these devices perform data transmission and reception over separate time/frequency resources. Depending on the way such resources are used, we can have either \ac{TDD} or \ac{FDD} operations, i.e., \ac{UL} and \ac{DL} transmissions occur over two different time or frequency resources, respectively. This approach has several advantages in terms of both ease of implementation and rather straightforward network operations to perform multi-cell transmissions. As a matter of fact, it can be argued that this implicitly sets a hard constraint on the spectral efficiency of the system. For this reason, many research efforts have been performed lately to investigate the potential and the feasibility of \ac{FD} communications, in which \red{the same time/frequency resource is used to perform the \ac{UL} and \ac{DL} data transmissions}. However, the strong \ac{SI} observed by the FD radio during the signal reception enforces a crucial obstacle to the feasibility of such approach. In other terms, a non-negligible portion of the transmitted signal is always received by the device's receive chain, in turn reducing the \ac{SINR} of the incoming useful signals \cite{Dua12,Bha13,Bha14}. This situation is schematically depicted in Fig.~\ref{fig:FD}.
\begin{figure*}[t!]
	\centering
	\includegraphics[scale=0.32]{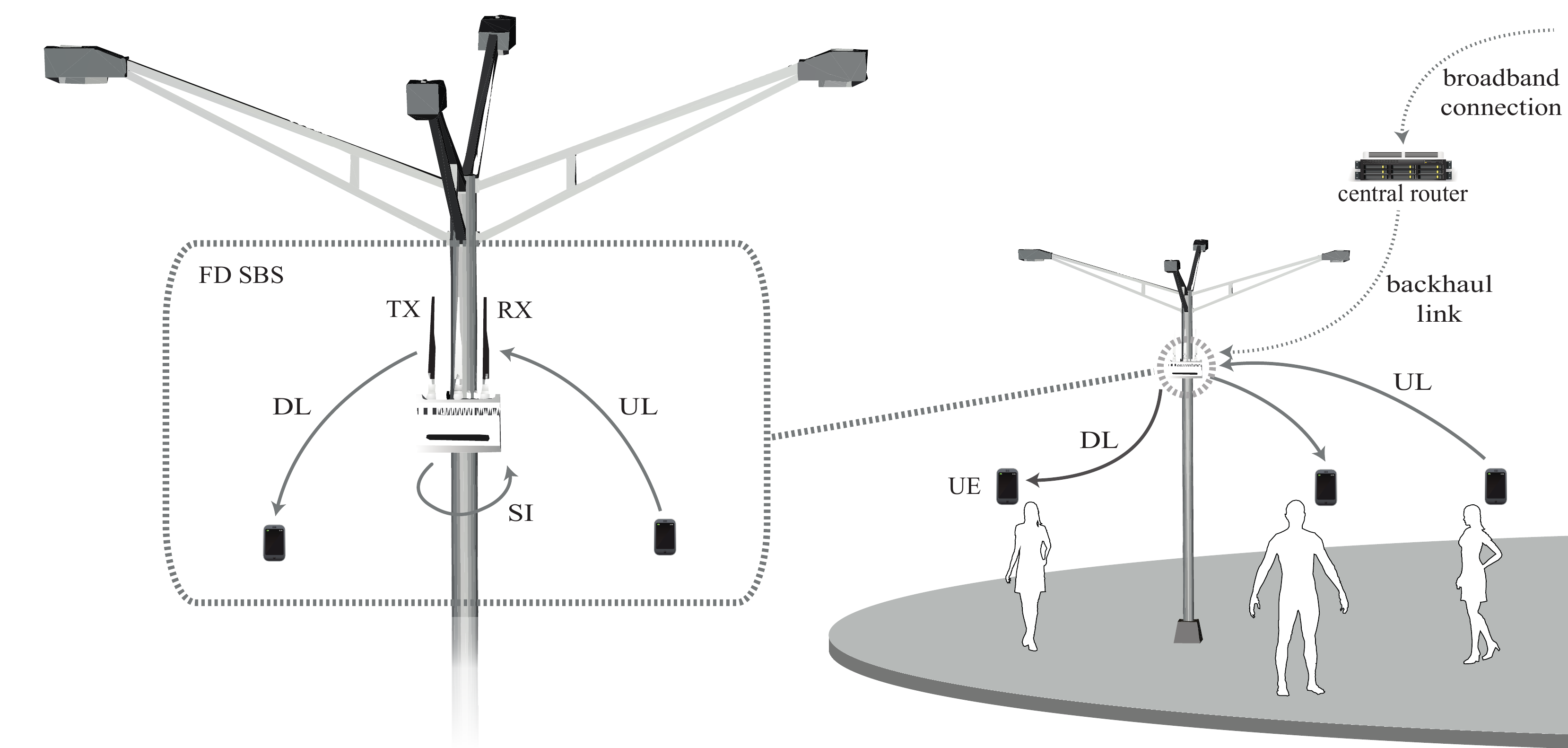}
	\caption{\ac{FD} \ac{SBS}: data transmission and reception \red{occur in the same time/frequency resource.}}
	\label{fig:FD}
\end{figure*}
Many different transceiver designs and \ac{SIC} algorithms have been devised to ensure the feasibility of \red{\ac{FD} operations} \cite{Bha13,Bha14,Cho10,Kno12,Bha14a,Phu13,Jai11,Li12,Ahm15,Kor14,Atz16}. These solutions can be classified into two major categories based on passive or active cancellation. In the former case, \ac{SIC} is achieved in the propagation domain by physical separation of the transmit and receive antennas. Conversely, active \ac{SIC} solutions exploit the \ac{FD} node's knowledge of its own transmitted signal to subtract it from the receive signal after appropriate manipulations and processing.

Unfortunately, the \ac{SI} is not the only problem that system designers must face when dealing with \ac{FD} radios. The major obstacle to their practical adoption in future \ac{5G} networks is arguably the aggregated interference footprint resulting from multiple and concurrent \ac{FD} communications within the network. Let us provide an example to highlight this issue. Consider a simple network composed \red{of several} \ac{FD} nodes arranged in \ac{BS}/\ac{UE} pairs and take an active \ac{BS}/\ac{UE} pair as reference. During \ac{UE}-to-\ac{BS} \ac{UL} operations, every neighboring \ac{BS} engaging in \ac{DL} transmission strongly interferes with the considered \ac{BS}, inducing the so-called \textit{\ac{BS}-to-\ac{BS} interference}. Similarly, during \ac{BS}-to-\ac{UE} \ac{DL} operations, all the \acp{UE} performing \ac{UL} transmission heavily interfere with the considered \ac{UE}, creating the so-called \textit{\ac{UE}-to-\ac{UE} interference}, also referred to as \ac{INI} \cite{Ale16,Atz16a}. In practice, the FD throughput gain tends to 2 in case of very sparse deployment of nodes. Nevertheless, such gain saturates quickly as the network density increases, the fundamental reason being that the number of interfering nodes also doubles with respect to the \ac{HD} case. This becomes more significant when either the link distance decreases or the node density increases \cite{Wan17}. In other words, \red{the theoretical throughput doubling brought by \ac{FD} at the device level does not seem to materialize straightforwardly at the network level} (regardless of the effectiveness of the adopted \ac{SIC} algorithms), unless specific and possibly costly countermeasures are taken. \red{As a result, the aggressive spatial frequency reuse inherent to dense network deployments may not be feasible due to the presence of a multitude of \ac{FD} links mutually interfering at all times}.

Studies and analysis of the \ac{FD} interference footprint have been recently carried out to identify viable strategies to reduce it and improve the scalability of the \ac{FD} throughput enhancement. A myriad of approaches has been proposed to address this problem; indeed they range from user scheduling algorithms to advanced interference management and power control techniques. A common feature shared by such solutions is that they require \red{ the adoption of additional signaling among nodes or the implementation of heavy infrastructural changes} \cite{Atz17,Ale16,Atz16a,Goy15,Wan15,Bai13,Atz16}. In this context, two fundamental results can be highlighted: on the one hand, it is shown in \cite{Goy15} that most of the theoretical network throughput gain is achievable if only the \acp{BS} operate in \ac{FD} while the \acp{UE} operate in \ac{HD}, and centralized scheduling decisions are taken by a central unit enjoying full access to global system information; on the other hand, it is shown in \cite{Bai13} that, in case of distributed network control and operations, \ac{FD} gains can be observed only if the \acp{UE} can exchange suitable information about \ac{INI} over one or more orthogonal control links. Hence, the relevance of the aforementioned results is mostly theoretical, as the advocated infrastructural changes at the network level are extremely expensive.
%Thus, if on the one hand the works in \cite{Goy15,Bai13} provide theoretical results which suggest that FD throughput could proficiently scale from device to network level, on the other hand they do not provide a practically feasible way to guarantee the realization of the very hard assumptions these results stem from.
One of the goals of this chapter is to investigate the feasibility of a constructive alternative to such approaches: this considers the interactions between \ac{FD} operations and smart caching strategies, and avoids any substantial changes to the network infrastructure and to the signaling exchange among nodes.

\begin{figure}[t!]
	\centering
	\includegraphics[scale=0.32]{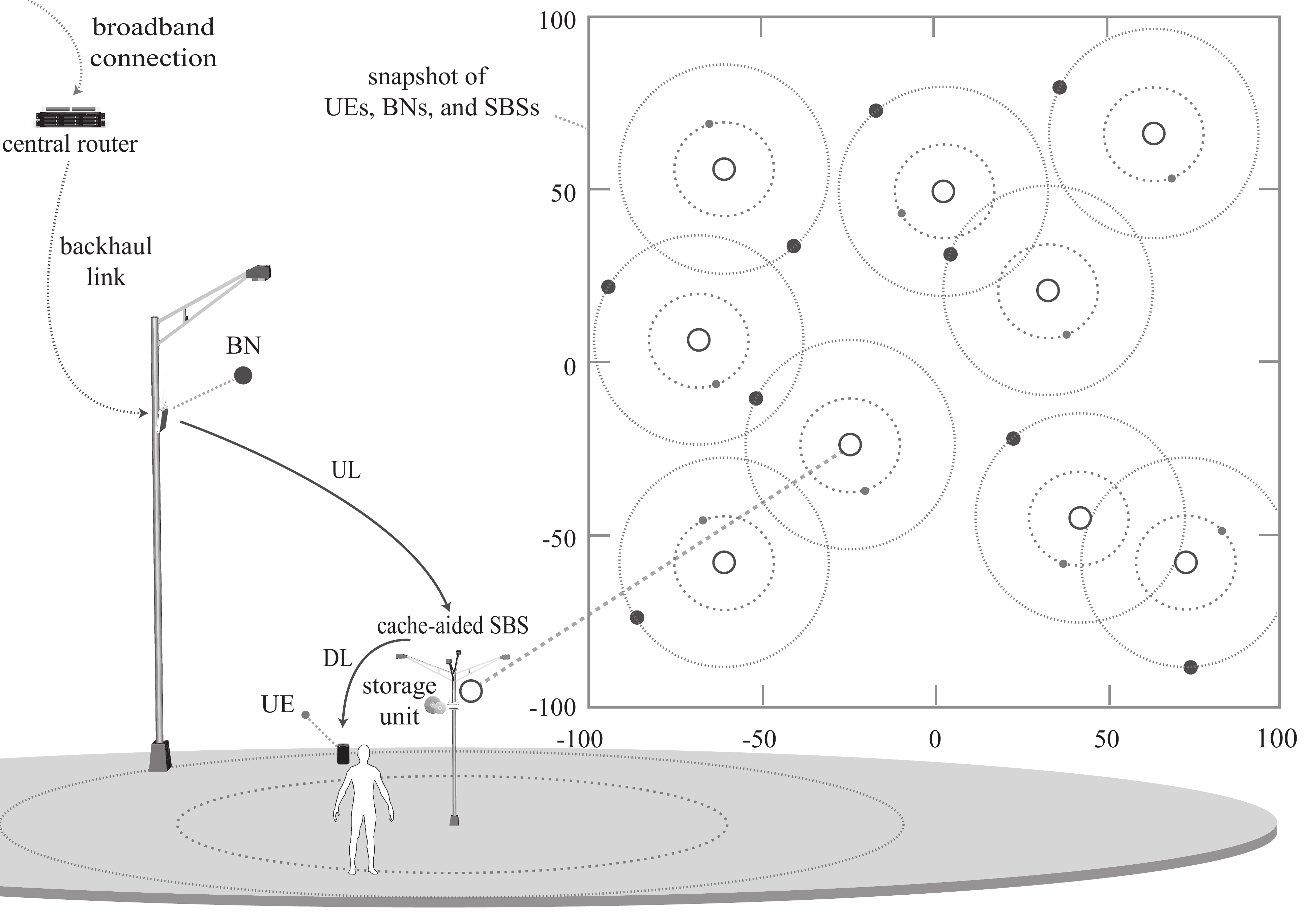}
	\caption{Snapshot of the marked PPP modeling the \acp{SBS}, \ac{UL} nodes, and \ac{DL} nodes.}
	\label{fig:snap}
\end{figure}

%=========================================================================
\section{System Model} \label{sec:SM}
%=========================================================================
\subsection{Network Model} \label{sec:SM_net}
%=========================================================================

Consider a \ac{UDN} comprising: \textit{i)} a tier of macro-cell \acp{BN} equipped with internet access, \textit{ii)} a tier of \acp{SBS} \red{providing} network coverage, and \textit{iii)} a set of mobile \acp{UE}. Each \ac{SBS} communicates with only one \ac{BN} in the \ac{UL} direction and transmits contents to only one \ac{UE} in the \ac{DL} direction, \red{functioning} as a relay between the two. The \acp{SBS} operate in \ac{FD}, whereas both \acp{BN} and \acp{UE} operate in \ac{HD} mode; the same time/frequency resource is used for \red{the communications in both directions}. In the following, and focusing our attention on the \acp{SBS}, the \acp{BN} and the \acp{UE} are referred to as \textit{\ac{UL} nodes} and \textit{\ac{DL} nodes}, respectively.

Spatial random models allow to \red{seize} the randomness of realistic ultra-dense small-cell deployments and, \red{in addition, to derive tractable and accurate} expressions for system-level performance \red{analysis} \cite{Hae12}. Therefore, we model the spatial distribution of the network nodes (i.e., \acp{SBS} and \ac{UL}/\ac{DL} nodes) \red{using} the homogeneous, independently marked \ac{PPP} $\Phi \triangleq \big\{ (x, u(x), d(x))\big\} \subset \mathbb{R}^{2} \times \mathbb{R}^{2} \times \mathbb{R}^{2}$. \red{Here}, we let $\Phi_{\sc} \triangleq \{x\}$ denote the ground \ac{PPP} of the \acp{SBS} with spatial density $\lambda$ (measured in [SBSs/m$^{2}$]), whereas the isotropic marks $\Phi_{\ul} \triangleq u(\Phi_{\sc}) = \{u(x)\}_{x \in \Phi_{\sc}}$ and $\Phi_{\dl} \triangleq d(\Phi_{\sc}) = \{d(x)\}_{x \in \Phi_{\sc}}$ denote the PPPs of the \ac{UL} and \ac{DL} nodes, respectively. Furthermore, let $r_{y,z} \triangleq \| y - z \|$ be the distance \red{between nodes} $y,z \in \Phi$; the distances of the \ac{UL} and \ac{DL} nodes from their associated \acp{SBS} are assumed fixed and are denoted by $R_{\ul} \triangleq r_{u(x),x}$ and $R_{\dl} \triangleq r_{x,d(x)}$, $\forall x \in \Phi_{\sc} $, respectively. It is thus evident that, according to these definitions, the \acp{PPP} $\Phi_{\ul}$ and $\Phi_{\dl}$ are dependent on the ground PPP $\Phi_{\sc}$ and have the same spatial density of the latter. Lastly, since the \acp{SBS} cover small areas compared with the \acp{BN}, one can reasonably assume that $R_{\ul} \gg R_{\dl}$. A snapshot of the considered two-tier network is given in Fig.~\ref{fig:snap}.

%=========================================================================
\subsection{Cache-aided Network Nodes} \label{sec:SM_caching}
%=========================================================================

Let us assume that the \ac{UL} nodes \red{have direct access to the} \textit{global file catalog} $\mathcal{F} \triangleq \lbrace f_{1}, f_{2}, \ldots, f_{F} \rbrace$, with $|\setF| = F$, which \red{can be interpreted as} a subset of all the contents available on the internet. Without loss of generality, we assume that all files have identical length, as files with different lengths can be always split into chunks of equal size. In this context, whenever a \ac{DL} node \red{sends a request for a content in} $\setF$, its serving \ac{SBS}, operating in \ac{FD} mode, fetches the corresponding file from the associated \ac{UL} node and delivers it to the \ac{DL} node. \red{In \ac{UDN} scenarios, however, the reliability of the content transmission} may be reduced by the aggressive spatial frequency reuse, which may sensibly diminish the throughput with respect to an equivalent \ac{HD} network.

\begin{figure}[t!]
	\centering
	\includegraphics[scale=0.4]{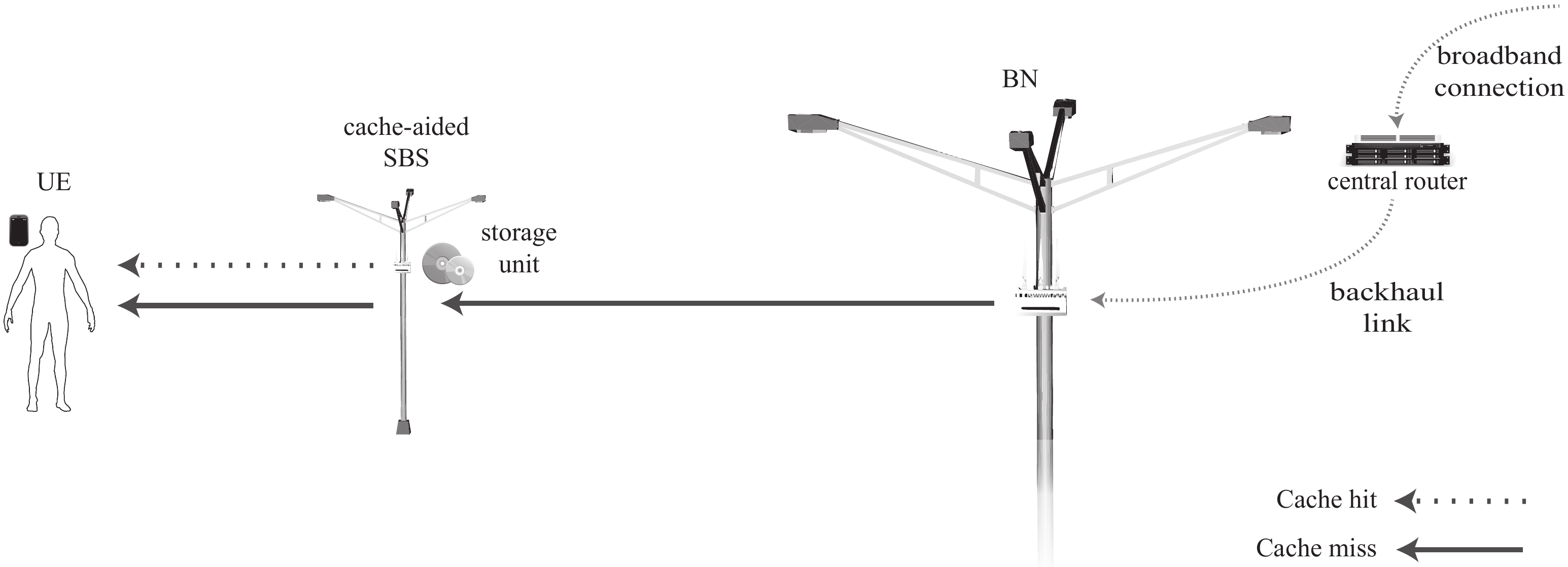}
	\caption{Communication links in case of cache hit and cache-miss.}
	\label{fig:cacheHitMiss}
\end{figure}

\red{Assume} that \ac{SBS} $x \in \Phi_{\sc}$ is equipped with a \textit{storage unit} $\Delta_{x}$ \red{with} size $S < F$ files and that \ac{DL} node $d(x)$ \red{sends a request for} file $f_{i} \in \setF$. Let  $\mathcal{P} \triangleq \lbrace p_{1}, p_{2}, \ldots, p_{F} \rbrace$, with $\sum_{i=1}^{F} p_{i} = 1$, be the set of request probabilities of each file, which depends on the files popularity over the whole network. Now, a \textit{cache-hit} event occurs whenever $f_{i} \in \Delta_{x}$, i.e., if $f_{i}$ is cached at \ac{SBS} $x$. In this case, \ac{DL} node $d(x)$ is served directly by \ac{SBS}~$x$ without any communication between \ac{SBS}~$x$ and \ac{UL} node $u(x)$; alternatively, a \textit{cache-miss} event occurs whenever $f_{i} \not\in \Delta_{x}$, i.e., if $f_{i}$ is not available in the cache, and \ac{SBS}~$x$ must fetch the file from \ac{UL} node $u(x)$ and deliver it to $d(x)$ in \ac{FD} mode (see Fig.~\ref{fig:cacheHitMiss}). Thus, a cache-hit event allows to offload the overlaying macro-cell infrastructure and, since the \ac{UL} becomes inactive, removes the need for the \ac{SBS} to operate in \ac{FD} mode. As a consequence, two major advantages can be observed in terms of \red{reduced interference}: \textit{i)} at the single-cell level, both the \ac{SI} (at the \ac{SBS}) and the \ac{INI} (at the \ac{DL} node) disappear; \textit{ii)} at the network level, the inter-cell interference is substantially reduced. Fig.~\ref{fig:systemmodel} provides a schematic representation of the so-obtained scenario, whose interference terms are described in Section~\ref{sec:SM_SIR}.

A key parameter to assess the effectiveness of the considered cache-aided approach is the so-called \textit{cache-hit probability}, denoted by $\Phit$, which is the probability that any file requested by a given \ac{DL} node is cached at its associated \ac{SBS}. The framework adopted in this chapter to model such probability is presented in Section~\ref{sec:CM}. In particular, such framework is designed to capture the local files popularity in non-cooperative random networks. Subsequently, we investigate the system-level performance gain brought by the deployment of cache-aided \acp{SBS} in \ac{FD} networks for a given $\Phit$ in Section~\ref{sec:PA}.

\begin{figure}[t!]
	\centering
%	{\def\svgwidth{\columnwidth}
%	\import{figures/}{fig_1_mod.pdf_tex}}
	\includegraphics[width=0.95\textwidth]{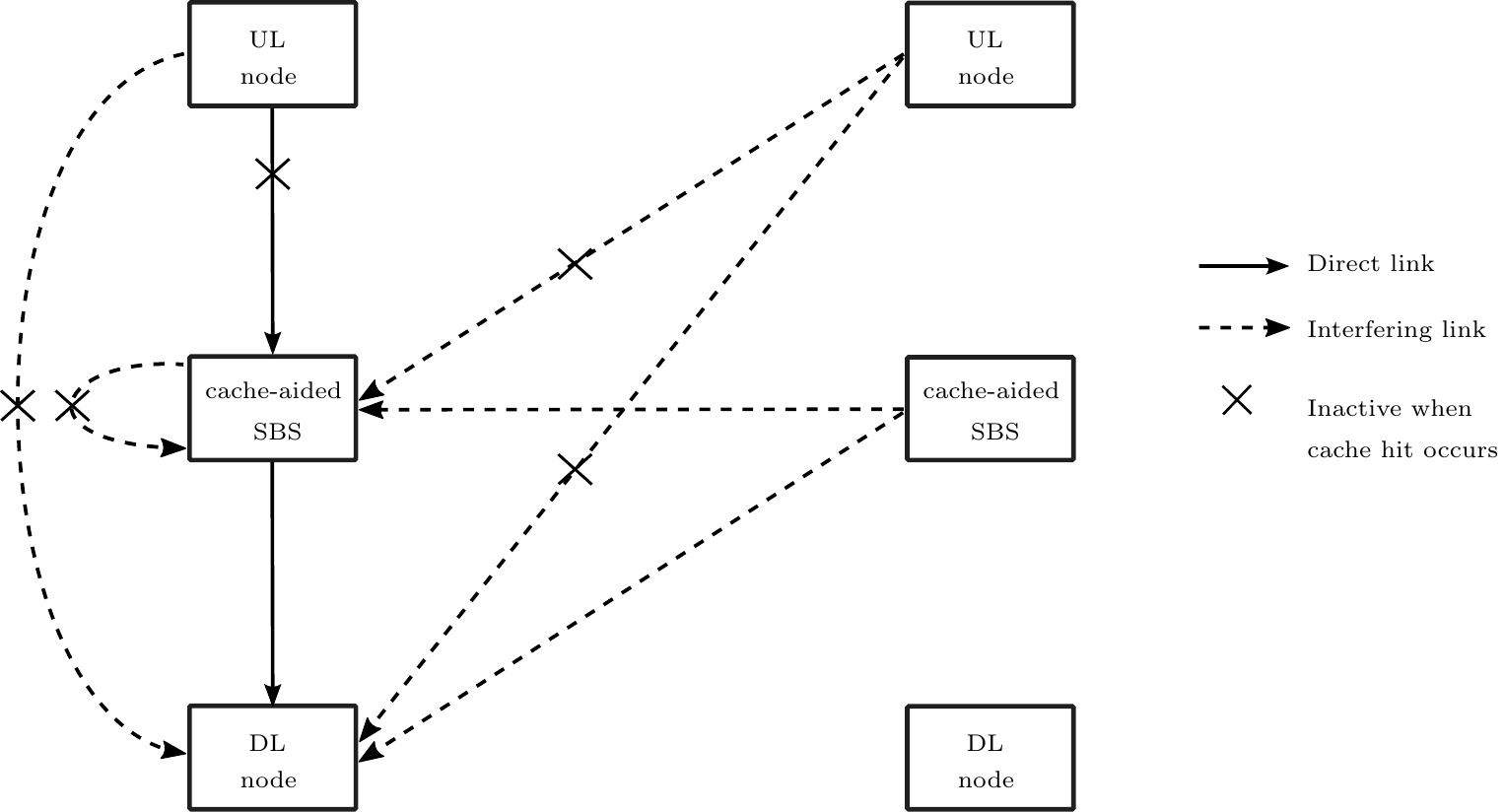} %\vspace{-3mm}
	\caption{Adopted system model including cache-aided \acp{SBS}, \ac{UL} nodes, and \ac{DL} nodes, with corresponding direct and interfering links.}
	\label{fig:systemmodel}
\end{figure}

%=========================================================================
\subsection{Channel Model} \label{sec:SM_ch}
%=========================================================================

In the considered system model, \red{we assume that all nodes are} single-antenna devices; the extension of our study to multi-antenna settings goes beyond the scope of this chapter and can be \red{accomplished using the analytical framework presented} in \cite{Atz17,Atz18}. In addition, it is assumed that the \ac{UL} nodes and the \acp{SBS} transmit with powers $\rho_{\ul}$ and $\rho_{\dl}$, respectively.

The wireless channel propagation is characterized as the combination of two main parameters, i.e., large-scale pathloss attenuation and small-scale fading. Let $\ell(y,z) \triangleq r_{y,z}^{-\alpha}$ be the pathloss function between nodes $y$ and $z$. We base our model upon the ITU-R urban micro-cellular (UMi) pathloss model described in \cite{3GPP36828}, where different attenuations are specified for the links between different types of nodes. Accordingly, we let $\alpha = \alpha_{2}$ if $y \in \Phi_{\ul} \land z \in \Phi_{\dl}$ (i.e., between \acp{BN} and \acp{UE}) and $\alpha = \alpha_{1}$ otherwise (i.e., between \acp{BN} and \acp{SBS} as well as between \acp{SBS} and \acp{UE}). In this respect, we assume non-line-of-sight propagation between \ac{UL} and \ac{DL} nodes, which results in stronger pathloss attenuation as compared to the other links, and set $\alpha_{2} \geq \alpha_{1} > 2$. Switching the focus to the small-scale fading, let $h_{y,z}$ denote the channel power fading gain between nodes $y$ and $z$. We assume that \red{the \ac{SI} channel is subject to Rician fading \cite{Dua12}, whereas all the other channels are subject to Rayleigh fading}. In other terms, we have that $h_{y,z} \sim \exp(1)$ if $y \neq z$ and $h_{y,y} \sim \Gamma(a,b)$. In particular, \red{the shape parameter $a$ and scale parameter $b$} of the SI distribution can be computed in closed form from the Rician $K$-factor $K$ and the \ac{SI} attenuation $\Omega$ measured at the \ac{SBS} when communicating in \ac{FD}, as detailed in \cite[Lem.~1]{Atz17}.

%=========================================================================
\subsection{Signal-to-Interference Ratio} \label{sec:SM_SIR}
%=========================================================================

Massive and dense small-cell deployments, such as the one considered in this chapter, are often characterized by heavy inter-cell interference as a result of the very short inter-site distance \cite{Atz17}. As a consequence, it is meaningful to specifically focus on the interference-limited regime, where the noise is overwhelmed by interference. In this context, the definition of an appropriate metric, such as the measured \ac{SIR} at the \acp{SBS} and at the \ac{DL} nodes, is paramount to be able to capture the essential features of the interference-limited regime. We start by denoting a cache-miss event at \ac{SBS} $x$ with the notation $\slashed{\Delta}_{x}$ and accordingly define the indicator function
\begin{align}
\mathbbm{1}_{\slashed{\Delta}_{x}} \triangleq \begin{cases}
1, & \mathrm{if} \ \slashed{\Delta}_{x}, \\
0, & \mathrm{otherwise}.
\end{cases}
\end{align}
The \ac{SIR} at \ac{SBS} $x$ may be written as
\begin{align} \label{eq:sir_x}
\sir_{x} & \triangleq \frac{\rho_{\ul} R_{\ul}^{-\alpha_{1}} h_{u(x),x}}{I_{x}}
\end{align}
with aggregate interference given by
\begin{align} \label{eq:I_x}
I_{x} \triangleq \sum_{y \in \Phi_{\sc} \backslash \{x\}} \big( \rho_{\dl} r_{y,x}^{-\alpha_{1}} h_{y,x} + \rho_{\ul} r_{u(y),x}^{-\alpha_{1}} h_{u(y),x} \mathbbm{1}_{\slashed{\Delta}_{y}} \big) + h_{x,x} \mathbbm{1}_{\slashed{\Delta}_{x}}
\end{align}
is the interference term. Similarly, the \ac{SIR} at \ac{DL} node $d(x)$ may be written as
\begin{align} \label{eq:sir_dx}
\sir_{d(x)} & \triangleq \frac{\rho_{\dl} R_{\dl}^{-\alpha_{1}} h_{x,d(x)}}{I_{d(x)}}
\end{align}
with aggregate interference given by
\begin{align} \label{eq:I_dx}
I_{d(x)} & \triangleq \sum_{y \in \Phi_{\sc} \backslash \{x\}} \big( \rho_{\ul} r_{y,d(x)}^{-\alpha_{1}} h_{y,d(x)} + \rho_{\dl} r_{u(y),d(x)}^{-\alpha_{1}} h_{u(y),d(x)} \mathbbm{1}_{\slashed{\Delta}_{y}} \big) + \rho_{\ul} r_{u(x),d(x)}^{-\alpha_{2}} h_{u(x),d(x)} \mathbbm{1}_{\slashed{\Delta}_{x}}.
\end{align}
The effect of equipping the \acp{SBS} with storage capabilities and \red{shifting} from a cache-free to a cache-aided scenario is rather evident upon observing \eqref{eq:I_x} and \eqref{eq:I_dx}. More precisely, a cache-hit event induces a reduction of the following major interference components, at both the network and the device level:
\begin{itemize}
	\item[$\bullet$] Aggregate network interference;
	\item[$\bullet$] \ac{SI} at the \acp{SBS} \cite{Dua12};
	\item[$\bullet$] \ac{INI} at the \ac{DL} nodes \cite{Ale16,Atz16a}.
\end{itemize}
We recall that the last two interference terms are the two main causes hindering the practical feasibility of \ac{FD} technology at the network level.

%=========================================================================
\section{Caching Model} \label{sec:CM}
%=========================================================================

A necessary step when performing studies on the performance of cache-aided networks is the definition of a caching model, whose role is to establish how files are requested and cached by \ac{DL} nodes and \acp{SBS}, respectively \cite{Bas15}. Accordingly, this section introduces the non-cooperative, static caching model used throughout this chapter, which aims at mimicking a geographical caching policy based on local files popularity.\footnote{More complex cooperative caching policies can be devised. However, this goes beyond the scope of this chapter.} In this regard, it is important to note that existing literature typically does not consider geographical aspects of the files popularity of the \acp{UE} when defining caching models (we refer to \cite{Pas16} for an overview on content request models).

Here, the spatial distribution of the contents from the global file catalog $\mathcal{F}$ is modeled by means of the homogeneous, independently marked \ac{PPP} $\Psi \triangleq \big\{ (y, f(y)) \big\} \subset \mathbb{R}^{2} \times \setF$, where $\Psi_{\setF} \triangleq \{ y \}$ \red{is} the \ac{PPP} of the files with spatial density $\eta$ (measured in [files/m$^{2}$]). \red{In this context}, each file $f_{i} \in \mathcal{F}$ \red{corresponds to} a thinned \ac{PPP} with spatial density $p_{i} \eta$. \red{Moreover}, we assume that the files in $\mathcal{F}$ are ordered by decreasing popularity, i.e., $p_{1} \geq p_{2} \geq \dots \geq p_{F}$.
%
%
%\begin{figure}[t!]
%	\centering
%	\includegraphics[scale=0.55]{figures/request_region.pdf}
%	\caption{Illustration of the request region.}
%	\label{fig:req_reg}
%\end{figure}
%
%
\red{The considered caching model consists of two core concepts}, i.e., the \textit{request region} and the \textit{caching policy}, which describe how \ac{DL} nodes request files and how \acp{SBS} cache files, respectively.
The introduction of a last notation is in order to be able to explicitly add a geographical dimension to these two concepts. Accordingly, we let $\mathcal{B} (z, \nu)$ denote the ball of radius $\nu$ (measured in [m]) centered at node $z \in \Phi_{\sc} \cup \Phi_{\dl}$.
\begin{definition}[Request region] \label{def:R}
\red{Assume} that \ac{DL} node $d(x) \in \Phi_{\dl}$ is interested in requesting \red{locally popular files. Then,} the request region of \ac{DL} node $d(x)$ is defined as
\begin{equation}
\mathcal{R}_{d(x)} \triangleq \big\{ \Psi_{\setF} \cap \mathcal{B}( d(x), R_{\mathrm{R}}) \big\}
\end{equation}
with $R_{\mathrm{R}}$ defined as the \textit{radius} of the request region.
\end{definition}

\begin{remarkb}
From a qualitative point of view, $R_{\mathrm{R}}$ is related to the local interests of the \acp{UE} with respect to globally requested files. In other terms, if \ac{DL} node $d(x)$ is interested in requesting all possible files in the global file catalog $\mathcal{F}$, then $R_{\mathrm{R}} \rightarrow \infty$ (provided that $\{p_{i} > 0\}_{i=1}^{F}$).
\end{remarkb}

\begin{definition}[Caching policy]\label{def:C}
\red{Assume} that \ac{SBS} $x \in \Phi_{\sc}$ is interested in caching \red{locally popular files. Then,} the potential cache region is defined as
\begin{equation}
\mathcal{C}_{x} \triangleq \big\{ \Psi_{\setF} \cap \mathcal{B} (x, R_{\mathrm{C}}) \big\}
\end{equation}
with $R_{\mathrm{C}}$ defined as the radius of the potential cache region. The caching policy of \ac{SBS}~$x \in \Phi_{\sc}$ is defined as
\begin{equation}
\Delta_{x} \triangleq \big\{ f_i : f_i \in \mathcal{C}_{x} \land i \leq S \big\}.
\end{equation}
\end{definition}

%\begin{figure}[t!]
%	\centering
%	\includegraphics[scale=0.55]{figures/cache_region.pdf}
%	\caption{Illustration of the potential cache region.}
%	\label{fig:cache_reg}
%\end{figure}

\begin{remarkb}
\acp{SBS} operating according to such caching policy will cache only geographically close (and, therefore, popular) files, in turn aiming at reducing the overhead associated with pre-fetching files from the \acp{BN}.
\end{remarkb}

\begin{remarkb}
Similarly to what has been previously observed for the request region, as $R_{\mathrm{C}} \rightarrow \infty$, we note that such caching policy will always converge to storing globally popular files as in {\rm \cite{Bas15}}.
\end{remarkb}

\noindent Finally, \red{the following Lemma formalizes the cache-hit probability under the described caching model.}

\begin{lemma} \label{lem:P_hit}
The cache-hit probability is \red{given by}
\begin{equation} \label{eq:P_hit}
\Phit = \frac{1}{F} \sum_{i=1}^{S} \big( 1-e^{- p_{i} \eta \pi R_{\mathrm{R}}^{2}} \big) \big( 1 - e^{-p_{i} \eta \pi R_{\mathrm{C}}^{2}} \big).
\end{equation}
\end{lemma}

\begin{proof}
%The null probability of a $n$-dimensional PPP of density $\nu$ in a volume $V \subset \Real^{n}$ is defined as the probability of no point falling into $V$ and is given by $e^{-\lambda V}$.
By definition, each file $f_i$ is distributed according to a thinned PPP with spatial density $p_i \eta$. Hence, we can straightforwardly infer that the probabilities of $f_i$ falling independently into the request region $\mathcal{R}_{d(x)}$ and into the potential cache region $\mathcal{C}_{x}$ are $1-e^{- p_{i} \eta \pi R_{\mathrm{R}}^{2}}$ and $1 - e^{-p_{i} \eta \pi R_{\mathrm{C}}^{2}}$, respectively. Assume now that $S \rightarrow \infty$, i.e., the \acp{SBS} are equipped with unlimited storage. Then, in this case, the cache-hit probability of file $f_i$ can be derived as the probability of file $f_i$ falling into both \red{the request region and the potential cache region}, which is readily given by $\big(1-e^{- p_{i} \eta \pi R_{\mathrm{R}}^{2}} \big) \big(1 - e^{-p_{i} \eta \pi R_{\mathrm{C}}^{2}} \big)$. Finally, considering \red{the totality of the contents included} in the global file catalog $\mathcal{F}$ and imposing storage constrains from Definition~\ref{def:C} yields the expression in \eqref{eq:P_hit}. \hspace{\fill} \qed
\end{proof}

\begin{remarkb}
Note that, in \red{non-cooperative caching settings, the maximization of $\Phit$ is straightforwardly achieved by caching the $S$ most popular files} at the \acp{SBS}. 
\end{remarkb}

%=========================================================================
\section{Performance Analysis} \label{sec:PA}
%=========================================================================

In this section, \red{we use tools from stochastic geometry to analyze the system-level performance enhancements brought by the considered cache-aided \ac{FD} network over its cache-free counterpart}. This choice provides analytical tractability of the problem and is crucial to guarantee the generality of our results. As main performance metric, we study the probability that a \ac{DL} node successfully receives a requested content, either through a direct transmission from its associated \ac{SBS} or with the aid of the corresponding \ac{UL} node. We term this metric as \textit{probability of successful transmission}, which is denoted by $\Psuc (\cdot)$. In this context, it is convenient to recall that the delivery of a requested file will be performed over different links depending on the occurrence of a cache-hit event. In particular:
\begin{itemize}
	\item[$\bullet$] \textit{Cache-hit event}: the transmission involves one hop, i.e., from the \ac{SBS} to the \ac{DL} node;
	\item[$\bullet$] \textit{Cache-miss event}: the transmission requires two hops, i.e., first from the \ac{UL} node to the \ac{SBS} and then from the latter to the \ac{DL} node through the \ac{SBS} (which \red{introduces} additional interference).
\end{itemize}
In our analysis, we focus on a \textit{typical \ac{SBS}}, indexed by $x$, and its marks $u(x)$ and $d(x)$, referred to as \textit{typical \ac{UL} node} and \textit{typical \ac{DL} node}, respectively. \red{Building on Slivnyak's theorem \cite[Ch.~8.5]{Hae12} and on} the stationarity of $\Phi_{\sc}$ (resp. of $\Phi_{\dl}$), the statistics of the typical \ac{SBS}'s (resp. of the typical \ac{DL} node's) signal reception are representative of the statistics seen by any \ac{SBS} (resp. by any \ac{DL} node) in the system.

Switching our focus back to $\Psuc (\cdot)$, we \red{consider} that a requested file is successfully received by the typical \ac{DL} node (i.e., through the two-hop communication link involving the typical \ac{UL} node, the typical \ac{SBS}, and the typical \ac{DL} node) if $\sir_{x} > \theta \land \sir_{d(x)} > \theta$, with $\theta$ defined as a target \ac{SIR} threshold. Additionally, we \red{consider} that the correct reception of the requested file over one hop uniquely depends on the \ac{SIR} experienced at the receiver, regardless of the considered hop. \red{For simplicity, and without loss of generality, we assume the same \ac{SIR} threshold for both \ac{UL} and \ac{DL} directions}.

Now, thanks to the caching capabilities at the typical \ac{SBS}, we can state that the \ac{UL} communication does not occur with probability $\Phit$. We can then express the probability of successful transmission as
\begin{align}
\label{eq:P_suc} \Psuc (\theta) \triangleq \Phit & \mathbb{P}(\sir_{d(x)} > \theta) + (1-\Phit) \mathbb{P}(\sir_{x} > \theta, \sir_{d(x)} > \theta).
\end{align}
Building upon this definition, other useful performance metrics can be expressed \red{in terms} of probability of successful transmission. Noteworthy examples are the outage probability, given by $\mathsf{P}_{\mathrm{out}} (\theta) \triangleq 1- \Psuc(\theta)$, and the achievable \ac{ASE}, defined as $\mathsf{ASE}(\theta) \triangleq \lambda \Psuc(\theta) \log_{2}(1 + \theta)$ (measured in [bps/Hz/m$^2$]).

Before proceeding with our analysis, we provide some useful preliminary definitions for the sake of notational simplicity in the remainder of the section: 
\begin{align}
\label{eq:upsilon_hat} \widehat{\Upsilon} (s) & \triangleq \frac{\pi (s \rho_{\dl})^{\frac{2}{\alpha_{1}}} \csc \big( \frac{2 \pi}{\alpha_{1}} \big)}{\alpha_{1}}, \\
\widetilde{\Upsilon} (s) & \triangleq \int_{0}^{\infty} \bigg( 1 - \frac{1}{1 + s \rho_{\dl} r^{- \alpha_{1}}} \Xi (s,r) \bigg) r \diff r, \\
\Xi (s,r) & \triangleq \frac{1}{2 \pi} \int_{0}^{2 \pi} \frac{\diff \varphi}{1 + s \rho_{\ul} (R_{\ul}^{2} + r^{2} + 2 R_{\ul} r \cos \varphi)^{- \frac{\alpha_{2}}{2}}}.
\end{align}
Recalling the expressions of $I_{x}$ and $I_{d(x)}$ in \eqref{eq:I_x} and \eqref{eq:I_dx}, respectively, a tight analytical lower bound on $\Psuc (\theta)$ is provided next in Theorem~\ref{th:P_sucLB}, with additional properties given in Corollary~\ref{cor:P_suc}.
\begin{theorem} \label{th:P_sucLB}
The probability of successful transmission in \eqref{eq:P_suc} is bounded as
$\Psuc (\theta) \geq \PsucLB (\theta)$, with
\begin{equation}
\label{eq:P_sucLB} \PsucLB (\theta) \triangleq \Phit \setL_{I_{d(x)}} (\theta \rho_{\dl}^{-1} R_{\dl}^{\alpha_{1}}) + (1 - \Phit) \setL_{I_{x}}^{\slashed{\Delta}_{x}} (\theta \rho_{\ul}^{-1} R_{\ul}^{\alpha_{1}}) \setL_{I_{d(x)}}^{\slashed{\Delta}_{x}} (\theta \rho_{\dl}^{-1} R_{\dl}^{\alpha_{1}})
\end{equation}
where $\setL_{I_{d(x)}} (s)$ is the Laplace transform of the interference observed at \ac{DL} node $d(x)$ in case of cache hit, whereas $\setL_{I_{x}}^{\slashed{\Delta}_{x}} (s)$ and $\setL_{I_{d(x)}}^{\slashed{\Delta}_{x}} (s)$ are the Laplace transforms of the interference observed  at \ac{SBS} $x$ and at \ac{DL} node $d(x)$, respectively, in case of cache-miss:
\begin{align}
\label{eq:L_dx} \mathcal{L}_{I_{d(x)}} (s) & \triangleq \exp \big( - 2 \pi \lambda \Phit \widehat{\Upsilon} (s) \big) \exp \big( - 2 \pi \lambda (1 - \Phit) \widetilde{\Upsilon} (s) \big), \\
\label{eq:L_x_D} \mathcal{L}_{I_{x}}^{\slashed{\Delta}_{x}} (s) & \triangleq \frac{1}{(1 + s \rho_{\dl} b)^{a}} \mathcal{L}_{I_{d(x)}}(s), \\
\label{eq:L_dx_D} \mathcal{L}_{I_{d(x)}}^{\slashed{\Delta}_{x}} (s) & \triangleq \Xi (s, R_{\dl}) \setL_{I_{d(x)}}(s).
\end{align}
\end{theorem}

\begin{proof}
The construction of \eqref{eq:P_sucLB} relies on the assumption of uncorrelated locations of the \ac{UL} and \ac{DL} nodes \red{in presence of a cache-miss event}. As a matter of fact, \red{according to} the Fortuin-Kasteleyn-Ginibre (FKG) inequality \cite[Ch.~10.4.2]{Hae12}, such uncorrelated case \red{yields} a lower bound on the network performance for the correlated case (we refer to \cite{Atz17} for further details). Therefore, given $\Psuc (\theta)$ in \eqref{eq:P_suc}, we can write
\begin{align} \label{eq:P_suc1}
\Psuc (\theta) \geq \Phit \mathsf{P}_{\mathrm{suc},2}(\theta) + (1-\Phit) \mathsf{P}_{\mathrm{suc},1}^{\slashed{\Delta}_{x}}(\theta) \mathsf{P}_{\mathrm{suc},2}^{\slashed{\Delta}_{x}}(\theta)
\end{align}
where $\mathsf{P}_{\mathrm{suc},2}(\theta)$ \red{represents} the probability of successfully transmitting a requested file from the typical \ac{SBS} to the typical \ac{DL} node in case of \red{a cache-hit event}, and $\mathsf{P}_{\mathrm{suc},1}^{\slashed{\Delta}_{x}}(\theta)$ (resp. $\mathsf{P}_{\mathrm{suc},2}^{\slashed{\Delta}_{x}}(\theta)$) denotes the probability of successfully transmitting a requested file from the typical \ac{UL} node to the typical \ac{SBS} (resp. from the typical \ac{SBS} to the typical \ac{DL} node) in case of \red{a cache-miss event}.

We begin by focusing on the latter components, i.e., $\mathsf{P}_{\mathrm{suc},1}^{\slashed{\Delta}_{x}}(\theta)$ and $\mathsf{P}_{\mathrm{suc},2}^{\slashed{\Delta}_{x}}(\theta)$. In particular, $\mathsf{P}_{\mathrm{suc},1}^{\slashed{\Delta}_{x}}(\theta)$ is obtained as the Laplace transform of $I_{x}$ in \eqref{eq:I_x} in presence of \ac{SI} \cite[Th.~1]{Atz17}, which is given by $\setL_{I_{x}}^{\slashed{\Delta}_{x}}(s)$ in \eqref{eq:L_x_D}. Likewise, $\mathsf{P}_{\mathrm{suc},2}^{\slashed{\Delta}_{x}}(\theta)$ is obtained as the Laplace transform of $I_{d(x)}$ in \eqref{eq:I_dx} in presence of \ac{INI}, which is given by $\setL_{I_{d(x)}}^{\slashed{\Delta}_{x}}(s)$ in \eqref{eq:L_dx_D}. Following a similar approach, $\mathsf{P}_{\mathrm{suc},2}(\theta)$ can be obtained as the Laplace transform of $I_{d(x)}$ in \eqref{eq:I_dx} in absence of INI, which is given by $\setL_{I_{d(x)}} (s)$. Now, \red{let us define $\widetilde{\Phi}_{\sc} \triangleq \{ x \in \Phi_{\sc} : \slashed{\Delta}_{x} \}$ and $\widehat{\Phi}_{\sc} \triangleq \Phi_{\sc} \backslash \widetilde{\Phi}_{\sc}$, which are, by definition,} independent \acp{PPP} with \red{spatial} densities $\Phit \lambda$ and $(1-\Phit) \lambda$, respectively. As a consequence, $\setL_{I_{d(x)}} (s)$ in \eqref{eq:L_dx} can be derived as
\begin{align}
\setL_{I_{d(x)}} (s) & = \Exp [e^{- s I_{d(x)}}] \\
& = \Exp \bigg[ \exp \bigg( - s \sum_{y \in \Phi_{\sc} \backslash \{x\}} \big( \rho_{\dl} r_{y d(x)}^{- \alpha_{1}} h_{y d(x)} + \rho_{\ul} r_{u(y) d(x)}^{- \alpha_{1}} h_{u(y) d(x)} \mathbbm{1}_{\slashed{\Delta}_{y}} \big) \bigg) \bigg] \\
& = \Exp \bigg[ \prod_{y \in \Phi_{\sc} \backslash \{x\}} \exp \bigg( - s \big( \rho_{\dl} r_{y d(x)}^{- \alpha_{1}} h_{y d(x)} + \rho_{\ul} r_{u(y) d(x)}^{- \alpha_{1}} h_{u(y) d(x)} \mathbbm{1}_{\slashed{\Delta}_{x}} \big) \bigg) \bigg] \\
\nonumber & = \Exp \bigg[ \prod_{y \in \widehat{\Phi}_{\sc} \backslash \{x\}} \exp \big( - s \rho_{\dl} r_{y d(x)}^{- \alpha_{1}} h_{y d(x)} \big) \bigg] \\ & \phantom{=} \ \times \Exp \bigg[ \prod_{y \in \widetilde{\Phi}_{\sc} \backslash \{x\}} \exp \bigg( - s \big( \rho_{\dl} r_{y d(x)}^{- \alpha_{1}} h_{y d(x)} + \rho_{\ul} r_{u(y) d(x)}^{- \alpha_{1}} h_{u(y) d(x)} \big) \bigg) \bigg]
%\label{eq:L_Id_1} \nonumber & = \Exp_{\widehat{\Phi}_{\sc}} \bigg[ \prod_{y \in \widehat{\Phi}_{\sc} \backslash \{x\}} \frac{1}{1 + s \rho_{\dl} r_{y d(x)}^{- \alpha_{1}}} \bigg] \\
%& \phantom{=} \ \times \Exp_{\widetilde{\Phi}_{\sc}} \bigg[ \prod_{y \in \widetilde{\Phi}_{\sc} \backslash \{x\}} \frac{1}{1 + s \rho_{\dl} r_{y d(x)}^{- \alpha_{1}}} \frac{1}{1 + \rho_{\ul} r_{u(y) d(x)}^{- \alpha_{1}}} \bigg]
\end{align}
and, using the moment-generating function of the exponential distribution, we obtain
\begin{align}
\nonumber \setL_{I_{d(x)}} (s) & = \Exp_{\widehat{\Phi}_{\sc}} \bigg[ \prod_{y \in \widehat{\Phi}_{\sc} \backslash \{x\}} \frac{1}{1 + s \rho_{\dl} r_{y d(x)}^{- \alpha_{1}}} \bigg] \\
\label{eq:L_Id_1} & \phantom{=} \ \times \Exp_{\widetilde{\Phi}_{\sc}} \bigg[ \prod_{y \in \widetilde{\Phi}_{\sc} \backslash \{x\}} \frac{1}{1 + s \rho_{\dl} r_{y d(x)}^{- \alpha_{1}}} \frac{1}{1 + \rho_{\ul} r_{u(y) d(x)}^{- \alpha_{1}}} \bigg].
\end{align}
Then, applying the probability generating functional of a PPP \cite[Ch.~4.3]{Hae12} yields
\begin{align}
\nonumber \setL_{I_{d(x)}} (s) & = \exp \bigg( - 2 \pi \lambda \Phit \int_{0}^{\infty} \bigg( 1 - \frac{1}{1 + s \rho_{\dl} r^{- \alpha_{1}}} \bigg) r \diff r \bigg) \\
\label{eq:L_Id_2} & \phantom{=} \ \times \exp \bigg( - 2 \pi \lambda (1 - \Phit) \int_{0}^{\infty} \bigg( 1 - \frac{1}{1 + s \rho_{\dl} r^{- \alpha_{1}}} \Xi (s,r) \bigg) r \diff r \bigg).
\end{align}
Finally, the integral appearing in the first exponential of \eqref{eq:L_Id_2} has a closed-form solution given by $\widehat{\Upsilon}(s)$ in \eqref{eq:upsilon_hat}. This concludes the proof. \hspace{\fill} \qed
\end{proof}

\begin{corollary} \label{cor:P_suc}
The lower bound on the probability of successful transmission in \eqref{eq:P_sucLB} is characterized by the following properties:
\begin{itemize}
\item[(a)] \hspace{1mm} $\PsucLB (\theta) \to \Psuc (\theta)$ as $\Phit \to 1$;
\item[(b)] \hspace{1mm} $\PsucLB (\theta) = \Psuc (\theta)$ in case of uncorrelated locations of the nodes between \ac{UL} and \ac{DL} communications.
\end{itemize}
\end{corollary}

\begin{proof}
The proof \red{follows directly} from Theorem~\ref{th:P_sucLB}. \hspace{\fill} \qed
\end{proof}

Lastly, we introduce the \textit{\ac{FD} throughput gain}, which will be used as a performance metric in Section~\ref{sec:NUM}, defined as
\begin{align}
\label{eq:TG} \mathsf{TG}_{\mathrm{FD}} (\theta) \triangleq 2 \Psuc (\theta) \exp \bigg( 2 \pi \lambda \frac{\pi \theta^{\frac{2}{\alpha_{1}}} \big( R_{\ul}^{2} + R_{\dl}^{2} \big) \csc \big( \frac{2 \pi}{\alpha_{1}} \big)}{\alpha_{1}} \bigg).
\end{align}
This performance metric quantifies the throughput gain of a cache-aided small-cell network operating in \ac{FD} mode as compared to its cache-free \ac{HD} counterpart by relating the probability of successful transmission $\Psuc (\theta)$ in the two settings and taking into account the \red{theoretical FD throughput doubling} (we refer to \cite{Atz17} for details). In particular, we note that the \ac{FD} setting outperforms its \ac{HD} counterpart when $\mathsf{TG}_{\mathrm{FD}} (\theta) > 1$.

\bgroup
\def\arraystretch{1.6}%
\begin{table}[!t]
	\caption{\red{System parameters} used in the simulations. \vspace{3mm}}
	\label{tab:simparams}
	\centering
	\normalsize
	\scalebox{0.8}{
		\begin{tabular}{|c|c|c|}
			\hline
			{\bf System Parameter}													& {\bf Symbol} 				& {\bf Value}			\\
			\hline
			\hline
			\hline
			Radius of request region 												& $R_{\mathrm{R}}$ 			& $8$ m 				\\
			Radius of potential cache region 										& $R_{\mathrm{C}}$			& $40$ m 				\\
			Catalog shape parameter 												& $\gamma$ 					& $0.7$ 				\\
			Storage-to-catalog ratio 												& $\kappa$ 					& $\{0.1, 0.35, 0.6\}$	\\
			\hline
			\hline
			Distance \ac{UL} node--\ac{SBS} 										& $R_{\mathrm{UL}}$  		& $20$ m 				\\
			Distance \ac{SBS}--\ac{DL} node 										& $R_{\mathrm{DL}}$  		& $5$ m 				\\
			Transmit power of \ac{UL} nodes   										& $\rho_{\mathrm{UL}}$  	& $30$ dBm 				\\
			Transmit power of \ac{DL} nodes   										& $\rho_{\mathrm{DL}}$ 	& $24$ dBm			 	\\
			Pathloss exponent \ac{UL} nodes--\acp{SBS}/\acp{SBS}--\ac{DL} nodes	& $\alpha_1$ 				& $3$			 		\\
			Pathloss exponent \ac{UL} nodes--\acp{DL} nodes						& $\alpha_2$				& $4$				 	\\
			Target \ac{SIR}   														& $\theta$ 					& $0$ dB 				\\
			\hline
			\hline
			Rician $K$-factor 														& $K$ 						& $1$ 				  	\\
			\ac{SI} attenuation 													& $\Omega$					& $60$ dB 				\\
			\hline
		\end{tabular}
	}
\end{table}
% Line width for the tables
\bgroup
\def\arraystretch{1.0}%

%=========================================================================
\section{Numerical Results and Discussion} \label{sec:NUM}
%=========================================================================

%\begin{figure}[!t]
%	\centering
%	\includegraphics[scale=1]{figures/Phit_VS_eta_markers.pdf}
%	\caption{The impact of file request density $\eta$ on the cache-hit probability.} \label{fig:Phit_VS_eta}
%\end{figure}
%
%Fig.~\ref{fig:Phit_VS_eta} plots the cache-hit probability $\mathsf{P}_{\mathrm{hit}}$ in \eqref{eq:P_hit} against the file request density $\eta$, where the obtained numerical results are accurately matched by the analytical expression derived in Lemma~\ref{lem:P_hit}. In particular, we observe that increasing the file request density yields a higher $\mathsf{P}_{\mathrm{hit}}$, until the cached files statistics converge to the file requests statistics. In other words, a higher file request density allows to capture the files popularity more accurately, which improves the efficiency of the storage use. Additionally, the impact of insufficient sampling of the files popularity is more evident when the \acp{SBS} are equipped with large storage units.

In this section, \red{we present and discuss numerical results obtained by means of suitable Monte Carlo simulations in order to assess the validity of our theoretical findings. We specifically focus on the analytical expressions obtained in Sections~\ref{sec:CM} and \ref{sec:PA}}.

\red{As commonly assumed in the literature (see, e.g., \cite{Pas16}), the global file catalog follows a Zipf popularity distribution such that the request probability $p_{i} \in \mathcal{P}$ of each file $f_{i} \in \mathcal{F}$ can be written as}
\begin{equation}
p_{i} =  \bigg ( i^{\gamma} \sum_{j=1}^{F} \frac{1}{j^{\gamma}} \bigg )^{-1}
\end{equation}
for a certain catalog shape parameter $\gamma$. Hence, the \acp{SBS} cache contents from the global file catalog (depending on the policy defined in Definition~\ref{def:C}) and serve the corresponding \acp{UE} accordingly. \red{The corresponding storage-to-catalog ratio is defined as $\kappa \triangleq \frac{S}{F} \leq 1$}. The values of the most relevant parameters adopted for the simulations are listed in Table~\ref{tab:simparams}; furthermore, \red{the shape parameter $a$ and the scale parameter $b$ of the \ac{SI} distribution, which appear in \eqref{eq:L_x_D}, are computed from the Rician $K$-factor $K$ and the \ac{SI} attenuation $\Omega$ measured at the \ac{FD} \acp{SBS}} as in \cite[Lem.~1]{Atz17}.

\begin{figure}[!t]
	\centering
	\includegraphics[scale=1]{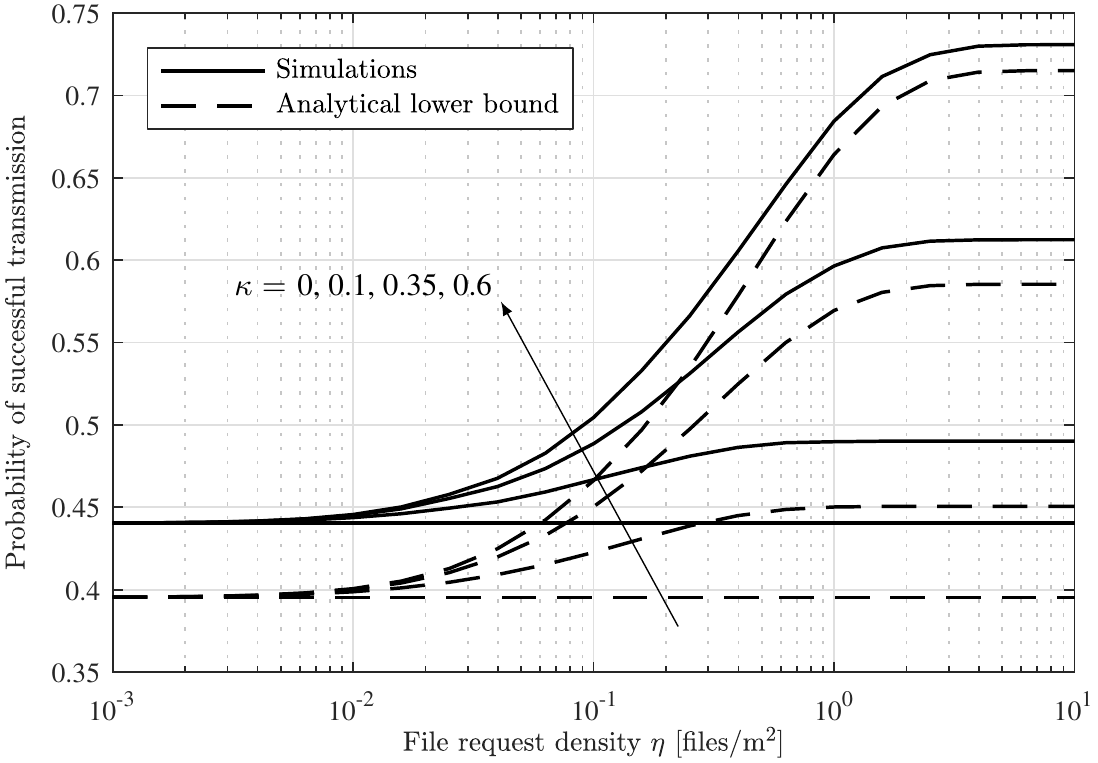}
	\caption{Probability of successful transmission against file request density, with \ac{SBS} density $\lambda=5 \times 10^{-4}$~SBSs/m$^{2}$.} \label{fig:Psuc_VS_eta}
\end{figure}

The probability of successful transmission $\Psuc (\theta)$ in \eqref{eq:P_suc} and its analytical lower bound $\PsucLB (\theta)$ in \eqref{eq:P_sucLB} are illustrated in Fig.~\ref{fig:Psuc_VS_eta} as functions of the file request density $\eta$ for a fixed \ac{SBS} density $\lambda=5 \times 10^{-4}$~SBSs/m$^{2}$. With reference to Corollary~\ref{cor:P_suc}, we recall that $\PsucLB (\theta)$ gives the exact expression of $\Psuc (\theta)$ \red{for uncorrelated locations of the nodes between \ac{UL} and \ac{DL} phases}. Moreover, note that the curves for $\kappa=0$ are related to the cache-free scenario analyzed in \cite{Atz17} (see also Fig.~\ref{fig:Psuc_VS_lambda} and Fig.~\ref{fig:TG_VS_lambda}). Qualitatively, it is evident from Fig.~\ref{fig:Psuc_VS_eta} that the probability of successful transmission grows with both the file request density $\eta$ and the storage-to-catalog ratio $\kappa$. On the one hand, \red{we observe that a higher file request density yields a larger} $\mathsf{P}_{\mathrm{hit}}$, which in turn improves the efficiency of the storage use. On the other hand, the variation experienced by the probability of successful transmission over $\eta$ increases with the storage capabilities at the \acp{SBS}, and so does the tightness of the analytical lower bound. Concerning this last aspect, we note that, even if such bound may look rather loose for $\kappa \leq 0.35$, its quantitative difference with the actual numerical performance never exceeds 10$\%$.

Assume now a file request density $\eta=1$~files/m$^{2}$. Fig.~\ref{fig:Psuc_VS_lambda} plots the probability of successful transmission $\Psuc (\theta)$ in \eqref{eq:P_suc} and its analytical lower bound $\PsucLB (\theta)$ in \eqref{eq:P_sucLB} as functions of the \ac{SBS} density $\lambda$. The analytical lower bound is remarkably tight and, in accordance with Fig.~\ref{fig:Psuc_VS_eta}, becomes increasingly accurate as the storage-to-catalog ratio $\kappa$ grows. Nonetheless, it is even more meaningful to analyze the \ac{FD} throughout gain in \eqref{eq:TG} together with its analytical lower bound (obtained by replacing $\Psuc (\theta)$ with $\PsucLB (\theta)$ in the aforementioned expression), which are illustrated in Fig.~\ref{fig:TG_VS_lambda} against the \ac{SBS} density $\lambda$. \red{In practice, higher \ac{ASE} can be achieved by deploying a very dense \ac{FD} network in which each \ac{SBS} is equipped with suitable caching capabilities. In this respect, we observe that:
\begin{itemize}
	\item a \ac{SBS} density $\lambda=10^{-4}$~SBSs/m$^{2}$ yields $\mathsf{TG}_{\mathrm{FD}} (\theta) = 1.7$ with $\kappa=0$ and $\mathsf{TG}_{\mathrm{FD}} (\theta) = 1.85$ with $\kappa=0.6$;
	\item a \ac{SBS} density $\lambda=10^{-3}$~SBSs/m$^{2}$ yields $\mathsf{TG}_{\mathrm{FD}} (\theta) = 0.42$ with $\kappa=0$ and $\mathsf{TG}_{\mathrm{FD}} (\theta) = 1.11$ with $\kappa=0.6$.
\end{itemize}
It is evident that the optimal tradeoff between the \ac{SBS} density and the storage size installed at each \ac{SBS} must be by found by network planners taking into account the deployment cost of each element; the interested reader may refer to \cite{Atz17a} for further details on this subject.}

\begin{figure}[!t]
	\centering
	\includegraphics[scale=1]{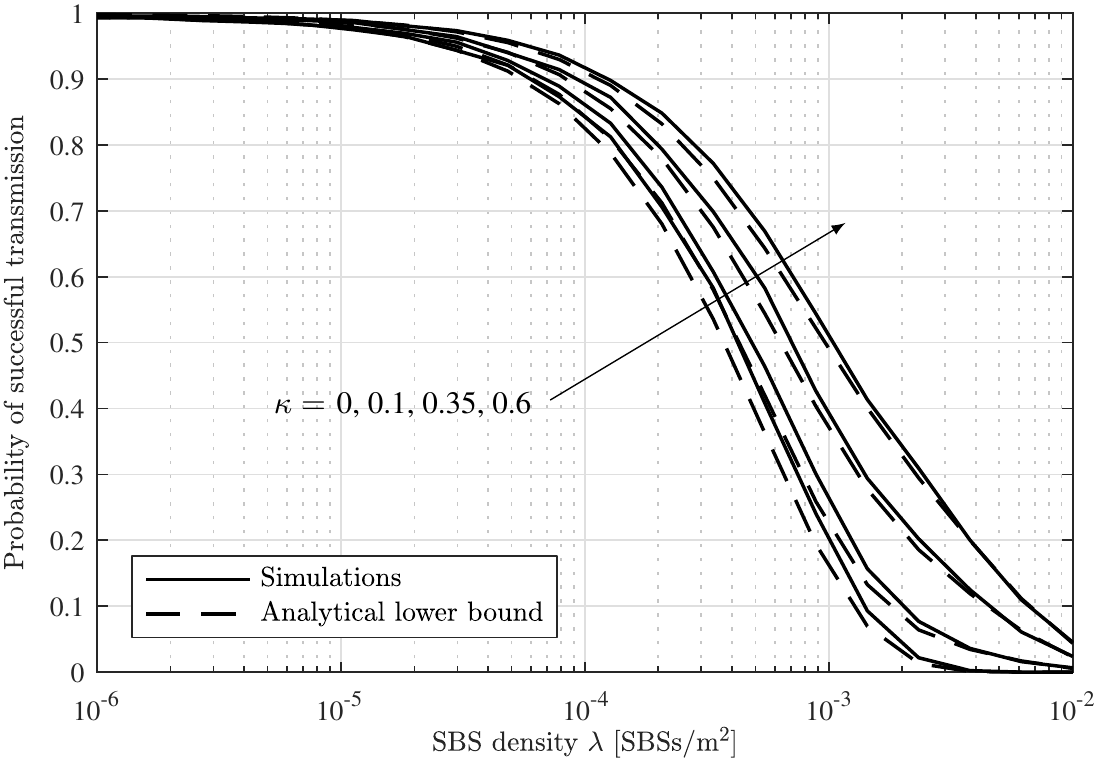}
	\caption{Probability of successful transmission against \ac{SBS} density, with file request density $\eta=1$~files/m$^{2}$.} \label{fig:Psuc_VS_lambda}
\end{figure}

%=========================================================================
\section{Conclusions} \label{sec:concl}
%=========================================================================

Several research efforts in both academic and industrial contexts have highlighted that edge caching can provide significant benefits in terms of network performance as, e.g., end-to-end access delay. Conversely, very few straightforward insights can be drawn on the benefits experienced at the physical layer when the network nodes are equipped with caching capabilities. This is mostly due to the complexity of the physical interactions occurring among devices in modern network.

This chapter takes a step forward with respect to the aforementioned position by showing that edge caching can actually offer a remarkable degree of interoperability with one of the most promising technologies for \red{next-generation network deployments}, i.e., \ac{FD} communications. More specifically, we show that integrating caching capabilities at the \ac{FD} \acp{SBS} is a cost-effective means of improving the scalability of the theoretical throughput doubling brought by the \ac{FD} paradigm from the device to the network level.

Our study considers an interference-limited \ac{UDN} setting consisting of several non-cooperative \acp{SBS} with \ac{FD} capabilities, which simultaneously communicate with both their served \acp{UE} and wireless \acp{BN}. In this case, the interference footprint of the \ac{UDN}, already significant by design, is further increased by the \ac{FD} operations. In fact, the latter induce higher levels of inter-cell and inter-node interference as compared to the \ac{HD} scenario, in turn causing a spectral efficiency bottleneck that prevents the theoretical \ac{FD} throughput doubling to occur at the network level. Fundamental results available in the literature show that most of such doubling can be achieved only if the network infrastructure is subject to radical and expensive modifications or if high-rate signaling is exchanged between \acp{UE} over suitable control links. 

\begin{figure}[!t]
	\centering
	\includegraphics[scale=1]{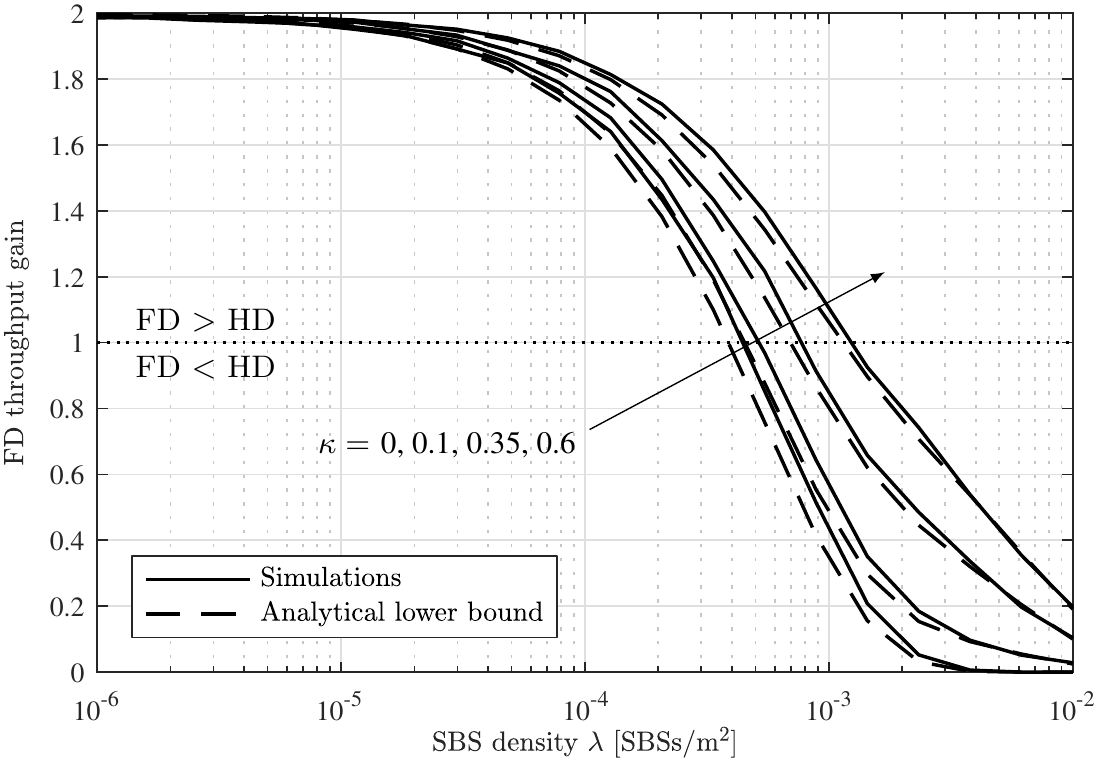}
	\caption{\ac{FD} throughput gain against \ac{SBS} density, with file request density $\eta=1$~files/m$^{2}$.} \label{fig:TG_VS_lambda}
\end{figure}

In this context, we add file storage capabilities to \acp{SBS} and \red{consider a geographical caching policy aiming at capturing local files popularity}, whereby the \acp{SBS} intelligently store popular contents anticipating the \acp{UE}' requests. The rationale of this choice is that the presence of pre-fetched popular files at the \acp{SBS} reduces the need for the latter to retrieve contents from the wireless \acp{BN} upon the \acp{UE}' request. This clearly diminishes the number of transmissions performed by the \acp{BN} towards the \acp{SBS}, in turn reducing the interference footprint of the \ac{UDN}. Remarkably, this low-cost solution can be implemented without the need for additional signaling between the nodes or any infrastructural change.

The performance of such cache-aided \ac{FD} network is characterized in terms of throughput gain as compared to its \ac{HD} counterpart. To this end, two fundamental metrics are identified and analyzed:
\begin{itemize}
	\item[$\bullet$] \red{The cache-hit probability, defined as the probability that any file requested by a given \ac{UE} from its serving \ac{SBS} is cached at the latter};
	\item[$\bullet$] The probability of successful transmission of a file requested by a \ac{UE}, either directly by its serving \ac{SBS} (if present in its cache) or by the corresponding \ac{BN}.
\end{itemize}
In particular, the second metric is used to derive an analytical lower bound on the throughput of the \ac{UDN}. As a final step, we perform a set of suitable numerical simulations to assess the performance enhancement brought by the adoption of edge caching. The obtained results highlight that \red{shifting} from cache-free to cache-aided \acp{UDN} allows to effectively operate the network in \ac{FD} mode while supporting higher \ac{SBS} densities, in turn improving the \ac{ASE}. In other words, the deployment of cache-aided \acp{SBS} has beneficial effects on the network throughput experienced over a given area, thanks to a non-negligible reduction of \red{the aggregate interference observed in the \ac{FD} network in comparison with the cache-free case}.

The results presented in this chapter demonstrate that the interoperability between edge caching and \ac{FD} communications is not only possible, but also desirable from the network throughput perspective. However, from a quantitative point of view, the extent of the benefits may strongly depend on several parameters, such as:
\begin{itemize}
	\item[$\bullet$] Adopted caching policy;
	\item[$\bullet$] \acp{UE} association policy;
	\item[$\bullet$] Mobility of the network nodes, either in the form of moving \acp{SBS} or classic \acp{UE}' dynamics.
\end{itemize}
Therefore, future additional studies and investigations should be performed in these directions to further deepen our understanding of the benefits brought by edge caching to the physical layer of wireless communication networks, especially when the \ac{FD} paradigm is adopted.

%=========================================================================
\section*{Acknowledgments} \label{sec:ack}
%=========================================================================

The work of Italo Atzeni was supported by the European Research Council under the Horizon 2020 Programme (ERC 670896 PERFUME).

\clearpage

\addcontentsline{toc}{chapter}{References}
\bibliographystyle{IEEEtran}
\bibliography{IEEEabrv,refs_caching}

\clearpage
\printindex

\end{document}